\documentclass[reqno]{amsart}
\usepackage{enumerate,amsmath,amssymb} 
\usepackage{pstricks,pst-node,pst-coil,pst-plot} 
\usepackage{hyperref}
\usepackage{color}
\theoremstyle{plain}

\newtheorem{thm}{Theorem}[section]
\newtheorem{theorem}[thm]{Theorem}
\newtheorem{cor}[thm]{Corollary}

\newtheorem{lem}[thm]{Lemma}

\newtheorem{prop}[thm]{Proposition}

\theoremstyle{remark}

\newtheorem{remark}[thm]{Remark}

\theoremstyle{definition}

\newtheorem{definition}[thm]{Definition}

\newcounter{mnotecount}[section]

\newcommand{\definedas}{\mathrel{\raise.095ex\hbox{\rm :}\mkern-5.2mu=}}

\def\epsilon{{\varepsilon}}
\def\phi{{\varphi}}

\let\<\langle 
\let\>\rangle

\newcommand{\bR}{\mathbb{R}}

\newcommand{\bS}{\mathbb{S}}

\newcommand{\cN}{\mathcal{N}}
\newcommand{\cL}{\mathcal{L}}

\newcommand{\gbar}{\overline{g}}

\renewcommand{\hbar}{\overline{h}}

\DeclareMathOperator{\divg}{div}

\newcommand{\ric}{\mathrm{Ric}}

\DeclareMathOperator{\hess}{Hess}

\begin{document} 
%%%%%%%%%%%%%%%%%%%%%%%%%%%%%%%%%%%%%%%%%%%%%%%%%%%%%%%%%%%%%%%%%%%%%%%%%

%opening
\title{\sc{On complete stationary vacuum initial data}}
\author{Julien Cortier}
\address{ETH Z\"urich\\ 
Forschungsinstitut f\"ur Mathematik\\ 
R\"amistrasse 101\\ 
8092 Zurich, Switzerland.} \email{jcortier@ihes.fr}
\author{Vincent Minerbe}
\address{Sorbonne Universit\'es, UPMC Univ Paris 06\\ 
UMR 7586, IMJ-PRG\\ 
4 place Jussieu, F-75005, Paris, France. } \email{minerbe@math.jussieu.fr}

\begin{abstract}
We describe a proof of M.T. Anderson's result \cite{And00a} on 
the rigidity of complete stationary initial data 
for the Einstein vacuum equations in spacetime dimension $3+1$, under an extra assumption 
on the norm of the stationary Killing vector field. The argument only involves
basic comparison geometry along with some Bochner-Weitzenb\"ock formula techniques.
We also discuss on the possibility to extend those techniques in higher dimensions.
\end{abstract}

\maketitle

%\tableofcontents

%%%%%%%%%%%%%%%%%%%%%%%%%%%%%%%%%%%%%%%%%%%%%%
\section{Introduction}
%%%%%%%%%%%%%%%%%%%%%%%%%%%%%%%%%%%%%%%%%%%%%%

In General Relativity, it is a natural task to try and classify  spacetime solutions of 
the Einstein equations under geometric requirements. Many basic questions are still wide open, 
even in the case of the \emph{vacuum Einstein equations} where the Ricci curvature tensor
of the spacetime metric vanishes.
However, significant progress has been done in particular cases, typically in presence of isometries. 
Among the simplest examples comes the study of \emph{spherically symmetric}, Ricci-flat spacetimes
in dimension $3+1$. The Birkhoff theorem asserts that such spacetimes are locally 
isometric to one of the maximally extended Schwarzschild spacetimes\cite{HE73}.

In this note, we restrict our attention to the class of spacetimes that are invariant 
under isometries in the time-direction. More precisely, we are interested here in spacetimes $(\cN,\gamma)$, 
solutions of the Einstein equations, in the special case of vanishing energy-momentum tensor and 
cosmological constant (hence Ricci-flat), which admit a timelike Killing vector field $\xi$. 
In order to avoid pathologies, we moreover assume that the orbits of this vector field are 
diffeomorphic to $\bR$, and that no closed timelike curves occur in the spacetime.
Such spacetimes are called \emph{stationary}; they are of considerable interest in 
General Relativity since they are expected to describe the final state of 
the gravitational collapse of a star into a black hole. We refer the interested reader to
\cite{Heu98} for a survey on stationary spacetimes.

A simple, but fundamental class of such spacetimes is the class of
\emph{static} spacetimes. 
These are stationary spacetimes $(\cN,\gamma)$ 
such that the orthogonal distribution with respect to the Killing vector field
$\xi$ is integrable.
An equivalent formulation is to say that $(\cN,\gamma)$ takes the form of a warped product
$$
\bR \times _u M := (\bR \times M\ ,\ - u^2 dt^2 + g )\;,
$$
where $M$ is a spacelike hypersurface of $\cN$ 
whose induced metric is the Riemannian metric $g$ and $u$ 
is a smooth, positive function on $M$. 
The fact that $\bR \times _u M$ is a Ricci-flat spacetime is equivalent to the fact
that the data $(M,g,u)$ satisfies the following conditions:
\begin{equation}\label{Static}
\begin{array}{rcl}
   \hess_g u & = & u g\\
   \Delta_g u & = & 0\;.
  \end{array}
\end{equation}
One also says that this static spacetime is \emph{vacuum}, which refers to the fact that
the energy momentum tensor of general relativity is chosen to be zero.

The problem of classifying \emph{static vacuum} spacetimes is therefore expressed as 
the problem of finding all positive solutions $(g,u)$ of the above system.
Fundamental examples of such static spacetimes are the Schwarzschild spacetimes. 
They have the expression
\[
g = \left(1 - \frac{2m}{r^{n-2}}\right)^{-1} dr^2 + r^2 \sigma _{\bS ^{n-1}}\ \ ,
\ \ u = \left(1 - \frac{2m}{r^{n-2}}\right)^{1/2}
\]
on the manifold $M = \left((2m)^{1/(n-2)},+\infty\right) \times \bS ^{n-1}$, 
where $m \in \bR$ is a parameter called the \emph{mass}. 
Some rigidity statements hold in spacetime dimension $n+1 = 4$.
For instance, Bunting and Masood-Ul-Alam were able to prove 
that Schwarzschild spacetimes are the only static vacuum ones which have the further 
property to be asymptotically flat \cite{BM87}.

In the more general setting of \emph{stationary vacuum} spacetimes, the classification
in dimension $3+1$ of the asymptotically flat ones and the uniqueness of \emph{Kerr}
spacetimes has been a major problem of mathematical relativity for the last decades. 
We will not develop further on this question and refer the reader to
\cite{CL08} and references therein.
In both cases, the spacetimes considered here may exhibit a black hole region and, as for the Schwarzschild and Kerr
examples, may fail to be geodesically complete.

Instead, we focus here on stationary vacuum spacetimes which are moreover complete.
The first rigidity result in this setting comes from Lichnerowicz \cite{Lic55}, 
under the further assumptions that the spacetimes considered are $3+1$ dimensional and asymptotically flat.
He obtains that only the Minkowski spacetime $\bR^{3,1}$ fulfills these properties (see also Einstein and Pauli \cite{EP43}).

Much more recently, Anderson \cite{And00a} proved the corresponding result 
without the asymptotic flatness assumption.

\begin{theorem}[Anderson, 2000]
 Let $(\cN,\gamma)$ be a 4-dimensional complete stationary vacuum spacetime. Then
 $(\cN,\gamma)$ is isometric to $(\bR \times M, -dt^2 + g)$, for some flat complete Riemannian manifold $(M,g)$.
\end{theorem}
The proof of this result in \cite{And00a} uses the full power of Cheeger-Fukaya-Gromov collapsing theory, with refinements specific 
to dimension three,  which makes it far from elementary. 

However, Case \cite{Cas10} (and subsequently Catino \cite{Cat12}) recently came back to the static vacuum setting and 
proved that all complete static vacuum  $n+1$-dimensional spacetimes $(\cN,\gamma)$ take the form of a product $(\bR \times M,-dt^2 + g)$,
where $(M,g)$ is a complete Ricci-flat $n$-dimensional Riemannian manifold. Their techniques are less 
sophisticated, relying on the Bochner formula, as well as comparison arguments \`a la Bakry-\'Emery.

In this paper, we will see how the same kind of techniques (and indeed without Bakry-\'Emery) can be adapted 
to provide a proof of rigidity in the stationary case, Theorem \ref{rigiditythm}, in dimension $n+1 = 4$ 
and under a suitable completeness assumption (instead of requiring the space-time to be complete, we assume a natural
metric on the orbit space is complete). The point is, even though our proof does not reach the full generality 
of Anderson's, it remains quite elementary. Note also that the stationary case is a bit more challenging than 
the static case, for the contribution of the non-trivial connection on the line bundle induces a contribution 
to the Ricci curvature which turns out to have a bad sign. This technicality is overcome by a conformal trick 
in dimension 3+1. In higher dimension, one can derive similar formulas for stationary initial data but they 
are harder to control. We discuss them at the end of the paper.

%%%%%%%%%%%%%%%%%%%%%%%%%%%%%%%%%%%%%%%%%%%%%%%%%%%%%%
\section{The setting, in dimension 3+1}\label{dim3}
%%%%%%%%%%%%%%%%%%%%%%%%%%%%%%%%%%%%%%%%%%%%%%%%%%%%%%

Definitions of stationary spacetimes existing in literature can vary depending on
the authors and the context,
although all of them assume the existence of a timelike Killing vector field
\footnote{Note however that this is no longer exact in the context of asymptotically flat
spacetimes with a black hole region, where the Killing vector field is usually asked to be 
timelike only in the asymptotic region, see \textit{e.g.} \cite{CL08}}.
We adopt the following definition in our work (compare with \cite{And00a} and 
\cite[Chap. XIV]{Cho09}).

\begin{definition}
A $(n+1)$-dimensional spacetime $(\cN,\gamma)$ is called \emph{stationary} if it has no closed
timelike curves and if there exists 
a timelike Killing vector field $\xi$ on $\cN$ whose orbits are complete. 
\end{definition} 
As mentioned in \cite{Har92}, the \emph{chronological} assumption,
corresponding to the non-existence of closed timelike curves, together with the orbit 
completeness prevent pathologies of the space of orbits.
\footnote{Without this assumption, an example of pathological spacetime is the 2-dimensional 
torus equipped with the Minkowski metric $-dx^2 + dy^2$. The orbits of the timelike Killing vector field 
$\xi = \sqrt{2}\; \partial _x + \partial _y$ are diffeomorphic to $\bR$, 
but the orbit space is not a smooth manifold.} 
In fact, a stationary spacetime $(\cN,\gamma)$ in the sense of the above definition 
can be seen as a principal $\bR$-bundle over the space of orbits $M$ which is a smooth manifold
diffeomorphic to any spacelike hypersurface of $\cN$ (see Geroch \cite{Ger67}). 

We will now see how one can characterize initial data corresponding to stationary vacuum spacetimes.
From the definition, a stationary spacetime is a $\bR$-principal bundle over a smooth base $M$:
\[
\pi : \cN \longrightarrow M.
\]
The fibers $\pi ^{-1}(\{p\})$, diffeomorphic to $\bR$, are the orbits of the timelike Killing 
vector field $\xi$, generator of the $\bR$ action. The orthogonal distribution determines a 
connection one-form $\theta$, $\xi$-invariant and with $\theta(\xi)=1$. The positive function $u$ 
defined by
\[
u^2 = - \gamma(\xi,\xi)
\]
is of course constant along the fibers, so we can think of it as a function on the base $M$. The spacetime metric then takes the form
\begin{equation}
\gamma = - u^2 \theta \otimes \theta + \pi^* g\;,
\end{equation}
where $g$ is the induced metric on the quotient space $M$.
We also denote by $\Omega := d\theta$ the corresponding curvature 2-form on $\cN$.
In dimension $n+1 = 4$, we also define the \emph{twist} 1-form as:
\[
\omega := - \frac12 u^3 \ast ^g \Omega\;,
\]
where $\ast ^g$ is the Hodge star operator associated with $g$.

We are interested here in the $3+1$-dimensional stationary spacetimes that moreover 
satisfy the Einstein vacuum equations, namely
\begin{equation}\label{vacuum}
 \ric ^{\gamma} = 0\;.
\end{equation}
The field equations obtained from \eqref{vacuum} on the data $(g,u,\omega)$ on $M$ 
then take the form 

\begin{equation}\label{Stat}
\left\{\begin{array}{rcl} \ric & = & u^{-1} \hess u + 2 u^{-4} (\omega \otimes \omega - |\omega|^2 g)\\
			\Delta u & = & -2u^{-3} |\omega|^2 \\
			\divg \omega & = & 3 \langle d \log u,\omega \rangle\\
			d \omega & = & 0 \;,
       \end{array} \right.
\end{equation}
where all the quantities are computed with respect to the metric $g$.
This system is obtained as the particular case $n=3$ of the computations performed in 
Section \ref{highdim}, see also \cite[pp. 455--456]{Cho09}.

This note offers a proof of the following statement:
\begin{theorem}\label{rigiditythm} 
 Let $(M^3,g,u,\omega)$ be a set of stationary vacuum initial data such that the metric
 $\gbar = u^2 g$ is complete. Then $u$ is a positive constant, $\omega = 0$ 
 and $(M,g)$ is flat.
\end{theorem}

An immediate consequence is the following.

\begin{cor}
 Let $(M^3,g,u,\omega)$ be a set of stationary vacuum initial data such that $g$ is complete and $u$ is bounded from below by a positive constant.
Then $u$ is a positive constant, $\omega = 0$ 
 and $(M,g)$ is flat.
\end{cor}

%%%%%%%%%%%%%%%%%%%%%%%%%%%%%%%%%
\section{The proof}
%%%%%%%%%%%%%%%%%%%%%%%%%%%%%%%%%

In order to prove, Theorem \ref{rigiditythm}, we use the so-called ``harmonic representation'' of static and stationary spaces \cite{And00a,Cat12} , which amounts
to rewriting the set of equations (\ref{Stat}) with respect to the conformal metric $\gbar = u^2 g$:
\begin{equation}\label{Statbar}
\left\{\begin{array}{rcl} \ric_{\gbar} & = & 2 d\log u \otimes d \log u + 2 u^{-4} \omega \otimes \omega\\
			\Delta_{\gbar} \log u & = & - 2 u^{-4} |\omega|_{\gbar} ^2 \\
			\divg_{\gbar} \omega & = &  4 \langle d \log u,\omega \rangle_{\gbar}\\
			d \omega & = & 0 \;.
       \end{array} \right.
\end{equation}

We see in particular that the Ricci tensor of $\gbar$ is non-negative. From now on, every notation and operator will refer to the metric $\gbar$. We will use the Bochner formula as follows: for a one-form $\alpha$, one has
\begin{equation}\label{Bochner}
\Delta |\alpha|^2 = 
2 |\nabla \alpha|^2 + 2 \ric (\alpha,\alpha) 
+ 2 \langle \alpha , \Delta_H \alpha \rangle\;,
\end{equation}
where $\Delta_H$ is the Hodge-de Rham Laplacian on differential forms, 
$\Delta_H = -(d d^* + d^* d)$ (our convention makes every Laplacian a nonpositive operator).

\begin{lem}
The function $v = \log u$ and one-form $\eta = 2 u^{-2} \omega$ satisfy
\begin{equation}\label{lambda34}
\Delta \big(|dv|^2 + \frac14 |\eta|^2 \big) \geq 
4 |dv|^4 + \frac14 |\eta|^4  + 2 \langle \eta,dv \rangle ^2\;.
\end{equation}
\end{lem}

\begin{proof}
We apply first the Bochner formula \eqref{Bochner} to $\alpha = dv$:
$$
\Delta |dv|^2 =
2 |\hess v|^2 + 2 \ric (\nabla v,\nabla v) 
+ 2 \langle \nabla v,\nabla \Delta v \rangle\;.
$$
Taking (\ref{Statbar}) into account yields
\begin{equation}\label{dv}
\Delta |dv|^2 =
2 |\hess v|^2 + 4 |dv|^4 + 4 u^{-4} \langle \omega , dv \rangle ^2 
+ 16 u^{-4} |\omega|^2 |dv|^2 - 4 u^{-4} \langle dv, d|\omega|^2 \rangle\;.
\end{equation}
Comparing with the static case \cite{Cat12}, we need to tackle the last term. In view of this,
we use  (\ref{Statbar}) to find $\Delta_H \omega = - d d^* \omega = d \divg \omega$ and 
then apply the Bochner formula \eqref{Bochner} to $\alpha = \omega$. This yields
\begin{equation}\label{om}
\Delta |\omega|^2 = 
8 \langle d\langle \omega, dv \rangle,\omega \rangle 
+ 2 |\nabla \omega |^2 + 4 \langle \omega, dv \rangle ^2 + 4 u^{-4} |\omega |^4\;.
\end{equation}
Now, we compute
\[
\begin{split}
\langle d\langle \omega, dv \rangle,\omega \rangle & =
\hess v (\omega,\omega) + \nabla \omega (\omega ,dv ) \\
&= \hess v (\omega,\omega) + \frac12 \langle dv, d|\omega|^2 \rangle\;,
\end{split}
\]
where the last equality is due to the fact that $d\omega = 0$, so
that $\nabla \omega$ is a symmetric tensor.
Inserting this into (\ref{om}) yields
\begin{equation}\label{om1}
\Delta |\omega|^2 =
8 \hess v (\omega,\omega) + 4 \langle dv, d|\omega|^2 \rangle 
+ 2 |\nabla \omega |^2 + 4 \langle \omega, dv \rangle ^2 + 4 u^{-4} |\omega |^4\;.
\end{equation}
In order to rewrite (\ref{om1}) in terms of $\eta$, 
we first note the formula for the Laplacian of a product:
$$
\Delta |\eta|^2 = 
4 u^{-4} \Delta |\omega|^2 -32 u^{-4} \langle dv, d|\omega|^2 \rangle 
+ 4 |\omega|^2 (-4 u^{-4} \Delta v + 16 u^{-4} |dv|^2)\;,
$$
hence
\[ \begin{split}
\Delta |\eta|^2 &= \\
& 32 u^{-4} \hess v (\omega,\omega) 
  - 16 u^{-4} \langle dv, d|\omega|^2 \rangle 
  + 8 u^{-4} |\nabla \omega|^2 \\
& + 16 u^{-4} \langle \omega, dv \rangle ^2 
  + 48 u^{-8} |\omega|^4 
  + 64 u^{-4} |dv|^2 |\omega|^2\;.
\end{split}
\]
Eventually, we find that for any parameter $\lambda \geq 0$, 
\begin{equation}\label{dvom}
 \begin{split}
\Delta (|dv|^2 + \lambda |\eta|^2 ) &= \\
& 2 |\hess v|^2 
  + 4 |dv|^4 
  + 4 u^{-4} (1 + 4 \lambda)  \langle \omega, dv \rangle ^2 \\
& + 16 u^{-4} (1+4\lambda) |dv|^2 |\omega|^2 
  - 4 u^{-4} (1 + 4 \lambda) \langle dv, d|\omega|^2 \rangle \\
& + 48 \lambda u^{-8} |\omega|^4 
  + 8 \lambda u^{-4} |\nabla \omega|^2 
  + 32 \lambda u^{-4} \hess v (\omega,\omega)\;.
 \end{split}
\end{equation}
Let us now replace $\omega$ by $u^2 \eta /2$:
\begin{equation}\label{dveta}
 \begin{split}
\Delta (|dv|^2 + \lambda |\eta|^2 ) &= \\
& 2 |\hess v|^2 + 4 |dv|^4 
  + (1 + 4 \lambda)  \langle \eta, dv \rangle ^2 
  + 4 (1+4\lambda) |dv|^2 |\eta|^2 \\
& - u^{-4} (1 + 4 \lambda) \langle dv, d(u^4|\eta|^2) \rangle 
  + 3 \lambda |\eta|^4 + 2 \lambda u^{-4} |\nabla (u^2 \eta)|^2 \\ 
& + 8 \lambda \hess v (\eta,\eta)\;.
 \end{split}
\end{equation}
To go one step further, we expand the term 
$$
d(u^4 |\eta|^2) 
= 4 u^4 |\eta|^2 dv + u^4 d |\eta|^2
$$
and the term 
$$
u^{-4} |\nabla (u^2 \eta)|^2 =
|\nabla \eta|^2 
+ 2 \langle dv, d|\eta|^2 \rangle 
+ 4 |dv|^2 |\eta|^2\;,
$$
so as to obtain
\begin{equation}\label{dveta1}
\begin{split} 
 \Delta (|dv|^2 + \lambda |\eta|^2 ) &= \\
 & 2 |\hess v|^2 + 4 |dv|^4 
   + (1 + 4 \lambda) \langle \eta, dv \rangle ^2 
   - \langle dv, d|\eta|^2 \rangle \\
   & +  \lambda \left[3 |\eta|^4 
   + 2 |\nabla \eta|^2 
   + 8 |dv|^2 |\eta|^2 
   + 8 \hess v (\eta,\eta)\right]\;.
 \end{split}
\end{equation}

Using the elementary lower bounds
$$
\hess v (\eta,\eta) \geq -\frac{a}{2} |\hess v|^2 - \frac{1}{2a}|\eta|^4\;,
$$
and
$$
-\langle dv, d|\eta|^2 \rangle = 
- 2 \nabla \eta ( \nabla v,\eta) \geq
- \frac{b}{2} |\nabla \eta |^2 - \frac{1}{2b}|dv^2||\eta|^2
$$
valid for any positive parameters $a$ and $b$, we find the inequality
\[
 \begin{split}
  \Delta (|dv|^2 + \lambda |\eta|^2 ) \geq & \\
      & 2 \left(1 - 2 a \lambda\right) |\hess v|^2 
      + 4 |dv|^4 
      + \lambda \left(3 - \frac{4}{a}\right) |\eta|^4 \\
      + & \left(8 \lambda - \frac{1}{b}\right) |dv|^2 |\eta|^2 
      +  \left(2\lambda  - b \right) |\nabla \eta |^2 
      + (1+4\lambda) \langle \eta, dv \rangle ^2 \;.
 \end{split}
\]
The choices $\lambda = 1/4$, $a = 2$ and $b = 1/2$ reduce this into  (\ref{lambda34}).
\end{proof}

\bigskip

Now, for any point $p$ and scale $R$, owing to $\ric \geq 0$, one can construct a smooth cutoff function $\chi_R \, : \, M \to [0,1]$ 
which is identically $1$ on  $B_R(p)$, vanishes outside $B_{2R}(p)$ and satisfies
$$
\frac{|d\chi_R|^2}{\chi_R} \leq c \, R^{-2}, \qquad |\Delta \chi_R| \leq c \, R^{-2},
$$  
for some universal constant $c$ (cf. \cite{CC96}  or the scaled version of theorem 8.16 in \cite{Che}; we indeed use the square of the cutoff function constructed there). 
We then consider the function $H$ defined by
$$
H = \chi_R \left(|dv|^2 + \frac14 |\eta|^2\right)\;,
$$
We can compute $\Delta H$ through the identity
\[
\begin{split}
\Delta H & = 
(\Delta \chi_R)\big(|dv|^2 + \frac14 |\eta|^2\big) \\ 
 & + \chi_R \Delta \big(|dv|^2 + \frac14 |\eta|^2\big) 
+ 2 \big\langle d\chi_R, d \big(|dv|^2 + \frac14 |\eta|^2\big) \big\rangle \;.
\end{split}
\]
In view of \eqref{lambda34}, at some point where $\chi_R>0$, we get
\[ \begin{split}
\Delta H \geq & \\
 & (\Delta \chi_R) \chi_R ^{-1} H 
 + \chi_R \left[ 4 |dv|^4 + \frac14 |\eta|^4  
	      + 2 \langle \eta,dv \rangle ^2 \right] \\ 
 & + 2 \chi_R ^{-1} \langle d \chi_R, dH \rangle 
 - 2 |d\chi_R|^2 \chi_R^{-2} H,
 \end{split}
\]
where
\[
4 |dv|^4 + \frac14 |\eta|^4 +  2 \langle \eta,dv \rangle ^2  \geq 
  2 \left(|dv|^2 + \frac14 |\eta|^2 \right)^2,
\]
so that
$$
\Delta H \geq 
  (\Delta \chi_R) \chi_R ^{-1} H 
  + 2 \chi_R ^{-1} H^2 
  + 2 \chi_R ^{-1} \langle d \chi_R, dH \rangle 
  - 2 |d\chi_R|^2 \chi_R^{-2} H \;.
$$
The compactly supported function $H$ admits a maximum at some point $p_0$ in $M$. If $H(p_0) >0$ , we have at $p_0$:
$$
0 \geq 
  (\Delta \chi_R) H 
  + 2 H^2 
  - 2 |d\chi_R|^2 \chi_R^{-1} H
$$
and thus
$$
H \leq |d\chi_R|^2 \chi_R^{-1} - \frac12 \Delta \chi_R \leq 2c \, R^{-2}.
$$
In particular, for any $R>0$, we get
$$
\sup_{B_R(p)} \left( |dv|^2 + \frac14 |\eta|^2 \right) \leq 2c \, R^{-2}.
$$
Letting $R$ go to infinity, we find that $dv$ and $\eta$ vanish, 
so that $u$ is constant, $\omega=0$, $g$ is Ricci-flat and therefore flat.
 \qed

%\begin{remark}
 %The construction of the above cut-off functions $\chi _R$ centered at a point $p$ remains valid despite the occurence of the cut-locus $C_p$ of $p$, from \cite{CC96,Che}.
 %Another way to deal with radial cut-off functions and the cut-locus issue consists in using a trick from Calabi, see \cite{Cal57}.
%\end{remark}

\begin{remark}
Instead of relying on \cite{CC96,Che}, we could have used the so-called Calabi trick \cite{Cal57}, which is maybe more elementary but somehow less transparent. 
\end{remark}

%%%%%%%%%%%%%%%%%%%%%%%%%%%%%%%%%%%%%%%%%%%%%%%%%%%%%%%%%%%
\section{Higher dimensional stationary data}\label{highdim}
%%%%%%%%%%%%%%%%%%%%%%%%%%%%%%%%%%%%%%%%%%%%%%%%%%%%%%%%%%%

In this section, we consider a principal $\bR$-bundle 
$
\pi : \cN \rightarrow M
$
over some smooth manifold $M^n$, $n \geq 3$, whose $\bR$ action is generated by the vector field $\xi$. 
We endow the total space $\cN$ with the Lorentzian metric 
$$
\gamma = - u^2 \theta ^2 + \pi ^* g,
$$
where $u$ is a positive function on $M$ and $\theta$ is 
a connection 1-form on $M$. So basically, $L_\xi \theta=0$, $\theta(\xi)=1$ and $u^2 = - \gamma(\xi,\xi)$.

We let $\Omega = d\theta$ be the curvature 2-form of the connection 1-form $\theta$,
and we denote by $\iota _. \Omega$ the mapping $ X \mapsto \iota _X \Omega$.
In particular, given an orthonormal frame $\{e_i\}_{i=1 \ldots n}$
and the dual coframe $\{e^i\}_{i=1 \ldots n}$, one has
\[
 \Omega = \sum_{1 \leq i < j \leq n} \Omega_{ij} e^i \wedge e^j 
        = \frac12 \sum_{i,j} \Omega_{ij} e^i \wedge e^j ,
\]
so that $|\iota _. \Omega|^2 = 2 |\Omega|^2$ and $(\divg \Omega)(X) = - \divg (\iota _X \Omega)$ for 
any vector field $X$.

The requirement that $(\cN,\gamma)$ is a solution of the vacuum Einstein's equations, 
i.e. Ricci-flat, yields the following
conditions on the data $(M,g,u,\Omega)$.

\begin{prop}
The Lorentzian manifold  $(\cN^{n+1},\gamma)$ determined by the data $(M,g,u,\Omega)$ as above
is Ricci-flat if and only if the following equations hold:
\begin{equation}\label{Statn}
\left\{\begin{array}{rcl} \ric & = & u^{-1} \hess u 
		- \frac12 u^{2} \langle\iota _.\Omega , \iota _.\Omega \rangle\\
			\Delta u & = & - \frac12 u^{3} |\Omega|^2 \\
			\divg \Omega & = & - 3\: \iota _{\nabla \log u}\Omega\\
			d \Omega & = & 0 \;.
       \end{array} \right.
\end{equation}

\end{prop}

\proof
We use the formalism of semi-Riemannian submersions, cf. \cite[Chap. 9]{Besse} and \cite{ONe66}. 
In particular, we denote $W,X,Y,Z$ for horizontal vectors, whereas $U := u^{-1}\xi$ 
is a unit vertical vector (in the sense that $\gamma (U,U) = -1 $). 
We denote by $D$ the Levi-Civita connection with respect to the metric $\gamma$ and
by $\nabla$ the one for the quotient metric $g$. 
The brackets $\langle . , . \rangle$ will refer to the metric $\gamma$.
We introduce the tensors $A$ and $T$ through their values
on vertical and horizontal vector fields:
$$
T_X U = 0\ ,\ T_X Y = 0\ ,\ T_U U = \mathcal{H}D _U U\ ,\ T_U X =
\mathcal{V} D _U X
$$
and
$$
A _U X = 0\ ,\ A _U U = 0\ ,\ A_X U = \mathcal{H}D _X U\ ,\ A_X Y =
\mathcal{V} D _X Y\;,
$$
where the operators $\mathcal{H}$ and $\mathcal{V}$ refer to the horizontal and 
vertical projection respectively.

We first estimate the above non-vanishing terms.
\begin{lem}
 For all horizontal vector fields $X$ and $Y$, the formulas hold:
 \begin{eqnarray}
  A _X Y = - \frac12 \Omega (X,Y) \xi\ &, &\ A_X U = - \frac12  u \iota _X \Omega\ , 
  \\
  T_U U =  \nabla \log u\ &, &\ T_U X = d \log u(X) U\ .
  \end{eqnarray}
\end{lem}

\begin{proof}
Let us first check that $A_X Y = \frac12 \mathcal{V} [X,Y]$. 

Indeed, for any horizontal $Z$, $A_Z Z$ is vertical and 
$\langle A_Z Z , U \rangle = - \langle Z, D_Z U \rangle = - \frac12 U (\langle Z,Z\rangle)$,
the last equality coming from the fact that $\pi_* [Z,U] = [\pi_* Z, \pi_* U] = 0$ so that 
$[Z,U] = D_Z U - D_U Z$ is vertical. We conclude that $A_Z Z = 0$ from the fact that $\langle Z,Z\rangle$
is constant along the (vertical) fibers, and we apply this to $Z = X+Y$, $Z=X$ and $Z=Y$ to get
that $A_X Y = - A_Y X$. The result now follows from the defining formula for $A_X Y$.

Next, we can write $A_X Y = \frac12 \mathcal{V} [X,Y] = \frac12 \theta ([X,Y]) \xi $,
and evaluate $\theta ([X,Y]) = - (\cL _X \theta)Y  = - \iota _X d\theta (Y)$ using the Cartan formula
and the property that $\theta$ vanishes on horizontal vectors. We have therefore obtained
$$
A _X Y = - \frac12 d\theta (X,Y) \xi\;.
$$
The formula for $A _X U$ now follows the fact that the tensor $A$ is alternate, 
in the sense that 
$\langle A_X U , Y \rangle = - \langle A_X Y , U \rangle$, and we obtain
$$
A_X U = - \frac12  u d\theta (X,.) =
- \frac12  u \iota _X d\theta\;.
$$

In order to establish the two remaining formulas concerning $T$, 
we claim that $T _{\xi} X = u^{-1} d u (X) \xi$ and $T_{\xi} \xi =  u \nabla u$.

Indeed, we compute $\langle T_{\xi} \xi , X \rangle = \langle D _{\xi} \xi ,X\rangle = 
\xi \langle \xi, X \rangle - \langle \xi , D _{\xi} X\rangle$.
In the meantime, the Lie bracket $[\xi,X]$ vanishes. Indeed, for the horizontal part, 
$ \pi _* [\xi , X] = [\pi _* \xi , \pi_* X] = 0$, whereas for the vertical part, 
$ \langle [\xi , X] , \xi \rangle =
\xi \langle X, \xi \rangle - (\cL _{\xi} \gamma)(\xi,X) - \langle X, [\xi,\xi] \rangle$,
which vanishes since $\xi$ is Killing for $\gamma$.

Hence, we can now write that $\langle T_{\xi} \xi , X \rangle =
- \langle \xi , D _X \xi \rangle $, which can be itself
expressed as $ - \frac12 X \langle \xi , \xi \rangle$. 
We eventually find the desired formula for $T_U U$ using the relation $U = u^{-1} \xi$.

The formula for $T_U X$ now follows from the property that $T$ is alternate in the sense
that $\langle T_U U, X \rangle = - \langle U, T_U X\rangle$.
\end{proof}
We can now use these formulas to compute the sectional curvatures, and then the Ricci 
curvature tensor of $\gamma$ evaluated on horizontal and vertical vectors. 
Let us first recall the O'Neill's formulas presented in \cite{ONe66,Besse}, in our setting 
where the fibers are one-dimensional:
\begin{equation}\label{ONeill}
 \begin{split}
  \langle R^{\gamma} (X,U)Y , U \rangle \ = &\ 
   \langle (D _X T)_U U , Y \rangle +
   \langle (D _U A) _X Y , U \rangle \\
   & - \langle T_U X , T_U Y \rangle +
   \langle A_X U , A_Y U \rangle ,
   \\
  \langle R^{\gamma} (X,Y)Z , U \rangle \ = &\
   \langle (D_Z A)_X Y , U \rangle +
   \langle A_X Y , T_U Z \rangle \\ 
   & - \langle A_Y Z , T_U X \rangle -
   \langle A_Z X , T_U Y \rangle ,
   \\
  \langle R^{\gamma} (X,Y)Z , W \rangle \ = &\
   \langle R (X,Y)Z , W \rangle -
   2 \langle A_X Y , A _Z W \rangle \\
   & + \langle A_Y Z , A _X W \rangle +
   \langle A_Z X , A _Y W \rangle
   .
 \end{split}
\end{equation}

Note that the convention used here (similarly to \cite{ONe66,Besse}) for the Riemann
curvature tensor is $R(X,Y) Z = \nabla _{[X,Y]} Z - \nabla _X \nabla _Y Z + \nabla _Y \nabla _X Z$,
as well as for $R^{\gamma}$ with respect to $D$.

We now rely on the formulas \eqref{ONeill} to derive the sectional curvatures $K^{\gamma}(X,U)$ and $K^{\gamma}(X,Y)$ for the metric $\gamma$, 
where $X,Y$ and $U$ satisfy $|X| = |X \wedge Y| =1$, $|U|^2 = -1 $:
\[
\begin{split}
K^{\gamma}(X,U) & = 
  \frac{\langle R^{\gamma}(X,U) X , U \rangle }{|X| ^2 |U| ^2} \\
%  & = -  \langle R^{\gamma}(X,U) X , U \rangle 
% \\
  & = 
  - \left[\langle (D _X T)_U U,X\rangle - |T_U X|^2 + |A_X U|^2\right]
\end{split}
\]
and
$$K^{\gamma} (X,Y) = K(X,Y) - 3 |A_X Y|^2\;,
$$
where $K$ is the sectional curvature related to the horizontal metric $g$, 
and, again, the symbols $\langle \ ,\ \rangle$ and $|.|$
refer to the metric $\gamma$. 

Note that the formula for $K^{\gamma}(X,U)$ differs from the one for Riemannian submersions 
in \cite[p241]{Besse} only by the factor $-1$.
In our setting, the formula for $K^{\gamma} (X,Y)$ takes the expression
\begin{equation}\label{KXY}
K^{\gamma}(X,Y) = K(X,Y) + \frac34  u^2 \Omega (X,Y)^2\;.
\end{equation}
Concerning $K^{\gamma}(X,U)$, we need to evaluate 
$$
\langle (D _X T)_U U,X\rangle =
X \langle T_U U, X\rangle - \langle T_{D _X U} U,X \rangle - 
\langle T_U (D _X U), X \rangle - \langle T_U U , D _X X \rangle\;.
$$
But $D _X U$ is horizontal (since $|U|^2$ is constant), hence 
$T_{D _X U} U = 0$.
For the same reason, $T_U(D _X U)$
is vertical, therefore $- \langle T_U (D _X U), X \rangle$ vanishes.
Then, $X \langle T_U U, X\rangle =  X.X. \log u = 
 \hess^{\gamma} \log u (X,X) + d \log u (D _X X) $, 
so that
$X \langle T_U U, X\rangle - \langle T_U U , D _X X \rangle =
 \hess^{\gamma} \log u (X,X)$.
All what remains now is:
\begin{equation}\label{KXU}
- K^{\gamma}(X,U) =
  \hess^{\gamma} \log u (X,X) 
+  |d\log u (X)|^2 + \frac14 u^2 |\iota _X \Omega|^2\;.
\end{equation}
We are now able to compute the component of the Ricci tensor of $\gamma$ from \eqref{KXY} 
and \eqref{KXU}.
Indeed, if $\{e_i\}_{i=1 \cdots n}$ is an orthonormal basis of $(M,g)$, one now has 
\[
 \ric^{\gamma} (U,U) = \sum_{i=1}^n \langle R^{\gamma} (e_i ,U) e_i , U \rangle ,
\]
hence
\begin{equation}\label{RicUU}
 \ric^{\gamma} (U,U) = -
\sum_{i=1}^n K^{\gamma}(e_i,U) = 
 \Delta^{\gamma} \log u + |d\log u |^2 + \frac14  u^2 |\iota _. \Omega|^2\;,
\end{equation}
where it is recalled that $\iota _. \Omega$ is the contraction mapping $ X \mapsto \iota _X \Omega $.
Since $u$ and therefore $\log u$ does not change along the flow of $\xi$, we have also that 
$\hess ^{\gamma} \log u (U,U) = 0$ and thus $\Delta ^{\gamma} \log u = \Delta \log u$, 
where it is recalled that $\Delta$ is the Laplacian for the metric $g$.
With the same observation and taking care of the signature of the metric $\gamma$, 
one obtains on horizontal vectors:
\[
\begin{split}
\ric^{\gamma} (e_i,e_i) &=
  - \langle R^{\gamma} (U,e_i) U , e_i \rangle 
  + \sum_{j \neq i} \langle R^{\gamma} (e_j,e_i) e_j , e_i \rangle \\
  &= K^{\gamma}(e_i,U) + \sum_{j\neq i} K^{\gamma} (e_i,e_j) \\
  &= \ric (e_i,e_i) - u^{-1} \hess u (e_i,e_i) 
  + \frac{1}{2} u^2 |\iota _{e_i} \Omega|^2\;,
\end{split}
\]
where $\hess$ is the Hessian with respect to the metric $g$. 
Thus, more generally, for $X,Y$ horizontal:
\begin{equation}\label{RicXY}
\ric^{\gamma} (X,Y) =  \ric (X,Y) - u^{-1} \hess u (X,Y) 
  + \frac{1}{2} u^2 \langle\iota _{X} \Omega,\iota _{Y} \Omega\rangle\;.
\end{equation}
We finally evaluate
\begin{equation}\label{RicUi}
\begin{split}
\ric^{\gamma} (e_i,U) &=  
   - \frac{1}{2} u 
    \left\{\divg (\iota _{e_i} \Omega) - 3 \iota _{\nabla \log u} \Omega (e_i)\right\}\\
    &= \frac{1}{2} u 
    \left\{(\divg \Omega)(e_i) + 3 \iota _{\nabla \log u} \Omega (e_i)\right\}.
\end{split}    
\end{equation}

Replacing $\ric ^{\gamma} = 0$ in \eqref{RicUU}, \eqref{RicXY} and \eqref{RicUi} 
gives the desired formulas. \qed

\begin{remark}
 As already mentioned in Section \ref{dim3}, this general result yields the field equations obtained 
 in dimension $3+1$ for the initial data $(M,g,u,\omega)$, providing that $\omega$ is chosen to be
 \[
  \omega = -\frac12 u^3 \ast ^g \Omega .
 \]
\end{remark}

\begin{remark} 
 Thanks to these formulas, one would expect that an analysis similar to the one in Section 
 \ref{rigiditythm} can be carried out.
 More precisely, when considering as above the metric $\overline{g} = u^2 g$ for $n \geq 4$, 
 and defining the function $ w = (n-3) \log u$, we get the identity
 \[
  \overline{\ric} + \overline{\hess} + \frac{n-5}{(n-3)^2} dw \otimes dw 
  = \frac{u^4}{2}\left(|\Omega|^2 _{\overline{g}} \overline{g} 
  - \langle \iota _. \Omega, \iota _. \Omega \rangle _{\overline{g}} ^2\right)\;.
 \]
A simple computation shows that the right-hand side of the above equation is non-negative.
On the other hand, for $n=4$, the left-hand side is the $1$-Bakry-\'Emery-Ricci tensor of $w$, 
denoted by $\overline{Ric}^1 _w$, whereas, for $n=5$, it is the $\infty$-Bakry-\'Emery-Ricci
tensor of $w$, $\overline{Ric}^{\infty} _w$. For larger $n$, there is no obvious way 
through any conformal change to rewrite the above identity in the form of a Bakry-\'Emery-Ricci tensor
with a non-negative right-hand side.
One might want at this point to use these facts 
(and the positive results of \cite{Cas10,Cat12}) to perform an analysis
as above for stationary vacuum solutions of dimension $n+1$ with $n = 4$ or $n=5$. 
But the generalisation of our proof of Theorem \ref{rigiditythm} 
to higher dimensions would require some more information on the curvature 2-form $\Omega$
or on its Hodge dual, as the generalised Weitzenb\"ock formulas 
involve the full Riemann tensor,
see \cite[Section 4]{Bou81} and \cite[Chap. 3]{Li}.
\end{remark}

\thanks{The authors wish to thank Piotr T. Chru\'sciel and Michael Eichmair for useful comments. 
J.C. acknowledges the IH\'ES (Bures-sur-Yvette) and the FIM, ETH (Z\"urich) for warm hospitality 
during part of this work.}

\bibliographystyle{amsplain}
\bibliography{../bibstatic}

\end{document}